\documentclass[twocolumn,pra,aps,amssymb,nofootinbib]{revtex4}
\usepackage{epsf,graphicx}
\usepackage{amsmath,amssymb,amsfonts}
\usepackage{multirow}
\usepackage{amsthm}
\usepackage{braket}
\usepackage{hyperref}
\def\equationautorefname~#1\null{Equation (#1)\null}
\def\figureautorefname~#1\null{Figure (#1)\null}
\usepackage{caption}
\usepackage{mathtools}

\begin{document}
\title{Quantum Error Correction using Hypergraph States}
\author{Shashanka Balakuntala$^*$ and Goutam Paul$^\#$}
\affiliation{$^*$Department of Applied Mathematics and Computational Sciences,\\
PSG College of Technology, Coimbatore 641 004, India,\\
Email: shbalakuntala@gmail.com\\
$^\#$Cryptology and Security Research Unit,\\
R. C. Bose Centre for Cryptology and Security,\\
Indian Statistical Institute, Kolkata 700 108, India,\\
Email: goutam.paul@isical.ac.in
}
\begin{abstract}
Graph states have been used for quantum error correction by Schlingemann et al. [Physical Review A 65.1 (2001): 012308]. Hypergraph states [Physical Review A 87.2 (2013): 022311] are generalizations of graph states and they have been used in quantum algorithms. We for the first time demonstrate how hypergraph states can be used for quantum error correction. We also point out that they are more efficient than graph states in the sense that to correct equal number of errors on the same graph topology, suitably defined hypergraph states require less number of gate operations than the corresponding graph states.
\end{abstract}
\maketitle
\newcommand{\pf}{{\bf Proof : }}
\newtheorem{definition}{Definition}
\newtheorem{theorem}{Theorem}
\newtheorem{lemma}{Lemma}
\newtheorem{remark}{Remark}
\newtheorem{corollary}{Corollary}

\section{\label{sec:intro}Introduction}
Error correction plays an important role in quantum information theory. Controlling operational errors and decoherence is a major challenge faced in the field of quantum computation. Quantum error correction was first proposed by P. W. Shor~\cite{shor1}. Since then there have been a lot of improvements in the field~\cite{robcal,knill,raylaf,dg1}. It is also to be noted that if one uses $n$ qubits to encode 1 qubit of information, one can detect no more than $n/2 - 1$ errors due to no cloning theorem~\cite{wootters}. Quantum error correction is essential, if one is to achieve fault-tolerant quantum computation that can deal not only with noise on stored quantum information, but also with faulty quantum gates, faulty quantum preparation, and faulty measurements~\cite{shor2,steane1,dg2}.

Schlingemann et al. gave a scheme for construction of error correcting codes based on a graph~\cite{werner}. The entanglement properties and applications of graph states can be found in~\cite{hein}. Verification of any error correction code is easier through graph states than the earlier existing methods. Other applications of graph states include secret sharing~\cite{markham}. Hypergraph states were first discussed in~\cite{rossi1,wang1}. In~\cite{rossi2}, they compute the amount of entanglement in the initial state of Grover's algorithm using hypergraph states. In the recent years hypergraph states are attracting a lot of attention, the local unitary symmetries of hypergraphs are discussed in~\cite{lyon}, their local Pauli stabilizers are discussed in~\cite{peters}. In~\cite{wang2}, they have discussed about the relationship between LME states and hypergraphs under local unitary transformations. Multipartite entanglement in hypergraphs are discussed in~\cite{wang3}. We, for the first time, use hypergraph states for error correction and show its efficiency over graph states. And the code that we discuss here are instances of stabilizer codes.

Consider a scenario where we have a network of CCTV cameras or satellites. This can be easily modelled as a graph, more precisely as a hypergraph. The relationship defined can be as follows: if two (or more) cameras cover the same area (from different angles) then we draw a hyper-edge between them. So now say the cameras qubits are entangled. If we are using these qubits to transmit a message and because of the fault in qubits there is an error in the transmitted message, this would mean that some of the cameras may be faulty and we need to repair them as soon as possible for smooth transmission of the feed. As the topology of the network follows a hypergraph, we try to give a condition through which we can easily identify the faulty bits and correct them. By finding the faulty bits, we get to know which cameras' qubits are faulty, and we can check the corresponding cameras. 

The paper is organized as follows. First we discuss the basics of hypergraph states, then we give the general condition for error correction. Next, we talk about the basic construction of the graphs which is required for the error correction. After that, we give the condition for error correction for a hypergraph state and a function to correct the error, followed by a short proof. Then we give the methods of encoding. Finally, we show what advantages hypergraphs have over graph states when using them for error correction.

\section{\label{sec:hg_intro}Hypergraph States}
Here we briefly review the concepts related to a hypergraph. As mentioned in~\cite{rossi1}, we first talk about the $k$-uniform hypergraphs. A $k$-uniform hypergraph $\Gamma^k = \{V,E\}$, containing $V$ a non-empty finite set of vertices $V$ and set of edges $E$, where each $e \in E$ connects exactly $k$ vertices.

Given a $k$-uniform hypergraph, one can find the $k$-uniform hypergraph state corresponding to the given graph as follows. Assign $\ket{+}$ to each and every vertex and if vertices $z_1, z_2, \ldots, z_k$ are connected by an edge and perform multi-qubit controlled-$Z$ operation between the connected qubits. So the final quantum state corresponding to the given $k$-uniform hypergraph will be as follows:
\begin{equation*}
\ket{\Gamma^{k}} = \prod_{\{z_1, z_2, \ldots, z_k\} \in E} C^{k}Z_{z_1, z_2, \ldots, z_k} \ket{+}^{\otimes |V|}.
\end{equation*}

where,
\begin{equation*}
    C^kZ = \textrm{diag}(\underbrace{1,1,\ldots,1}_{2^k-2},1,-1)
\end{equation*}

Next we discuss general hypergraph states. A hypergraph $\Gamma^{\le n} = \{V,E\}$ is a non-empty set $V$ of $n$ vertices and a set $E$ of edges, where the edges can connect any number of vertices from 1 to $n$. Given a hypergraph, the corresponding state can be written as follows:
\begin{equation*}
\ket{\Gamma^{\le n}} = \prod_{k = 1}^{n}\prod_{\{z_1, z_2, \ldots, z_k\} \in E} C^{k}Z_{z_1, z_2, \ldots, z_k} \ket{+}^{\otimes |V|}.
\end{equation*}

\section{\label{sec:err_corr}Error Correction}
General characteristics of quantum error correction was discussed by E. Knill and R. Laflamme in \cite{knill}. In quantum information, error correcting code is defined to be an isometry\footnote{An isometry is a transformation which is invariant with respect to the distance, i.e., the distance is preserved before and after transformation.} $v : \mathcal{H} \rightarrow \mathcal{K}$, where $\mathcal{H}$ is the input Hilbert space and $\mathcal{K}$ is the output Hilbert space with $dim(\mathcal{K}) > dim(\mathcal{H})$, so that input density operator $\rho$ is transformed by this coding into $v\rho v^\ast$, a density operator in $\mathcal{K}$. The code's output is then sent through the channel. The channel noise is described by a certain class of errors belonging to a linear subspace $E$ of operators on $\mathcal{K}$. Thus the channel can be represented as a completely positive linear map as follows:
\begin{equation*}
T(\sigma) = \sum_{\alpha}F_\alpha \sigma F_\alpha^\ast ,
\end{equation*}
where $\sigma \in \mathcal{K}$, $F_\alpha \in E$. The isometry $v$ is a error correcting code for $E$, if there is a recovery operator $R$ such that:
\begin{equation}
\label{cond_g}
R(T(v\rho v^\ast)) = \rho , \quad\qquad \forall \rho \in \mathcal{H}.
\end{equation} 

Here, the recovery operation is such that it has to act as an inverse to two things, first it has to reverse the operation done by the error transformation $T$, and after doing this we should be able to find the density operator $\rho$ which was initially sent.

Next we will discuss the basic construction of hypergraph states in order to prepare them for error correction and we will arrive at a condition for error correction using the hypergraph states. This construction is similar to that of~\cite{werner}.

\section{\label{sec:basic_construction}Error Correction Using Hypergraph States}
The following determines the codes that we construct:
\begin{description}
	\addtolength{\itemindent}{1.5cm}
	\item[$\bullet$] An undirected graph $\Gamma$ with two kinds of vertices: the set $X$ of input vertices and the set $Y$ of output vertices. If the vertices $z_1,z_2,\ldots,z_k \in \{X \cup Y\}$ are related then $\Gamma(z_1,z_2,\ldots,z_k) = 1$.
The notation $\Gamma^k$ is for a $k$-uniform hypergraph and $\Gamma^{\le n}$ is for a general hypergraph.
	\item[$\bullet$] A finite \textit{abelian group} $G$ with a non-degenerate symmetric \textit{bicharacter}.
\end{description}
In~\cite{werner}, the bicharacter function was $G \times G \rightarrow \mathbb{C}$, here we extend it to $k$-uniform bicharacter function. 

\theoremstyle{definition}
\begin{definition}{[$k$-uniform bicharacter function]}
A $k$-uniform bicharacter is a function $\chi^k: G \times G \times \cdots \times G \rightarrow \mathbb{C}$, which takes $k$ input arguments and maps it to a complex number such that, $\chi^k(g_1+h,g_2,\ldots,g_k) = \chi^k(g_1,g_2,\ldots,g_k)\chi^k(h,g_2,\ldots,g_k), \forall h$. The non-degeneracy is in a sense that:
\begin{equation*}
\sum_{g = g_1,g_2,\ldots,g_{k-1}}\chi^k(g,g^\prime) = |G|\delta(g^\prime),
\end{equation*}
\end{definition}
where,
\begin{equation*}
 \delta(g^\prime) = 
\begin{cases} 
1, & \mbox{if }g^\prime = 0,  \\
0, & \mbox{if }g^\prime \ne 0.
\end{cases}
\end{equation*}
The combined input system is described in $|X|$-fold tensor product space\footnote{$L^2(G)$ generates the Hilbert Space, which consists of all functions $\psi : G \rightarrow \mathbb{C}$, with inner product $\langle\phi,\psi\rangle = \int \overline{\phi}(g) \phi(g) dg$, where $\overline{z}$ is the complex conjugate of $z$.} $\mathcal{H}^{\otimes X} = L^2(G^X)$. A vector in this space is denoted by the notation $g_z, \forall z \in X$. The entire collection of these $g_z$ vectors is denoted as $g^X$. The same holds for the output system too. Now we define the isometry for $k$-uniform hypergraphs. 
\begin{equation*}
v_\Gamma^k : L^2(G^X) \rightarrow L^2(G^Y).
\end{equation*}
The scalar product in this space is given by,
\begin{equation*}
(v_\Gamma^k \psi)(g^Y) = \int v_\Gamma^k [g^{X\cup Y}]\psi(g^X) dg^X.
\end{equation*}
The $v_\Gamma$ under the integral denotes the kernel of the operator $v_\Gamma$ which depends on both input and output vertices. Explicitly the expressions for this kernel was given in~\cite{werner}, and for the $k$-uniform hypergraph it is given by:
\begin{equation*}
v_\Gamma^k [g^{X\cup Y}] = |G|^{-1}\prod_{z_1,z_2,\dots,z_k} \chi^k(g_{z_1},g_{z_2},\ldots,g_{z_k})^{\Gamma^k(z_1,z_2,\dots,z_k)}.
\end{equation*}
Now this concludes the construction of the isometry. We will give the condition of the error correction after we give more generalized form of this. Lets see how we can generalize these functions to a general hypergraph.
The isometry is defined by:
\begin{equation*}
v_\Gamma^{\le n} : L^2(G^X) \rightarrow L^2(G^Y).
\end{equation*}
Similarly the scalar product here is given by,
\begin{equation*}
(v_\Gamma^{\le n} \psi)(g^Y) = \int v_\Gamma^{\le n} [g^{X\cup Y}]\psi(g^X) dg^X.
\end{equation*}
And the kernel's expression is defined as follows:
\begin{equation*}
v_\Gamma^{\le n} [g^{X\cup Y}] = \prod_{k=1}^{n}\prod_{z_1,z_2,\dots,z_k} \chi^k(g_{z_1},g_{z_2},\ldots,g_{z_k})^{\Gamma^{\le n}(z_1,z_2,\dots,z_k)},
\end{equation*}
where it is normalized by $|G|^{-1}$. 

We say a vertex $z \in \mathcal{N}(x)$ if there is an edge which connects $z$ and $x$, which is same as saying $\Gamma^{\le n}(z,x,\ldots,z^\prime) = 1$.\\
We define another indicator function $f$ as follows:
\begin{equation*}
f(x,y) = 
\begin{cases} 
1, & \mbox{if }x \in \mathcal{N}(y),  \\
0, & \mbox{if }x \notin \mathcal{N}(y).
\end{cases}
\end{equation*}
Now we define a homomorphism which will help us to define the condition for error correction easily. This existed for graph states and now we try to define it for the hypergraphs.
A homomorphism between two subsets $K$ and $L$ of $X \cup Y$ denoted by $H: G^L \rightarrow G^K$ is defined as follows: 
\begin{equation*}
H_L^K(g^L) = \left(\sum_{l \in L} g_lf(l,k)\right)_{k \in K} .
\end{equation*}
Using this homomorphism, the necessary and sufficient condition for the error correction can be derived.

\begin{theorem}
Given a finite abelian group $G$ and an undirected hypergraph $\Gamma^{\le n}$, an error configuration $E \subset Y$ is detected by the quantum code $v_\Gamma$ if and only if the system of equations:
\begin{equation}
\label{suf1}
H_{X \cup E}^I g^{X\cup E} = 0	\qquad with\ I = Y\setminus E
\end{equation}
implies:
\begin{equation}
\label{suf2}
g^X = 0 \quad \textrm{and} \quad H_E^X g^E = 0.
\end{equation}
\end{theorem}
\begin{proof}
To prove this claim, we first see the equivalence of \autoref{cond_g} and the general condition mentioned by \cite{knill}. This was given by Schlingemann et al. in~\cite{werner} which reduces to the following:
\begin{equation}
\label{red}
\begin{aligned}[b]
    v_\Gamma^\ast F v_\Gamma[g^X, h^X]
   = \int dg^E dg^I dh^E \quad \overline{v_\Gamma[g^X, g^E, g^I]} \\ \times F[g^E,h^E] v_\Gamma[g^X, h^E, g^I]
\end{aligned}
\end{equation}
by choosing $ F = \ket{g^E}\bra{h^E}$, which can be further factorized as:
\begin{equation}
    \label{nsc}
    w_{[\Gamma, E]}[g^{X \cup E}, h^{X \cup E}] = C(g^E, h^E) \delta(g^X - h^X).
\end{equation}

We follow the proof structure as in~\cite{werner}. We define another function for any two $K$ and $K^\prime$ of $X \cup Y$,
\begin{equation*}
    \chi^{\Gamma^{\le n}}(g^K, g^{k^\prime})= \prod_{k = 1}^n\prod_{z = z_1,z_2,\ldots,z_{k-1},z_k} \chi^k(z)^{\Gamma^k(z)},
\end{equation*}

where $\{z_1,z_2,\ldots,z_{k-1}\} \in K$ and $z_k \in K^{\prime}$ and the product is over all possible $k-1$ elementary subsets in $K$ with its combination of 1 element in $K^\prime$. So the expression inside the integral of \autoref{red} can be expressed as:
\begin{equation}
    \label{red1}
    |G|^X \quad \frac{\chi^{\Gamma^{\le n}}(h^{X \cup E}, h^{X \cup E})}{\chi^{\Gamma^{\le n}}(g^{X \cup E}, g^{X \cup E})} \frac{\chi^{\Gamma^{\le n}}(h^{X \cup E}, g^{I})}{\chi^{\Gamma^{\le n}}(g^{X \cup E}, g^{I})},
\end{equation}

where $I = Y \setminus E$. To carry out the integral with respect to one variable $g_i, i \in I$, the $g_i$ dependent part of $\chi^{\Gamma^{\le n}}(h^{X \cup E}, g^{I})$ is as follows:
\begin{equation*}
    \prod_{k = 1}^{n}\prod_{\{z, i\}, z \in X\cup E} \chi^k(h_z,g_i)^{\Gamma^{\le n}(z,i)},
\end{equation*}

where $z=\{z_1,z_2,\ldots,z_{k-1}\}$ and $h_z=\{h_{z_1},\ldots,h_{z_{k-1}}\}$. Similarly the same for the factor $\chi^{\Gamma^k}(g^{X \cup E}, g^{I})$. Thus the $g_i$ dependent part of \autoref{red1} is:
\begin{equation*}
    \prod_{k = 1}^{n}\prod_{\{z, i\}, z \in X\cup E} \chi^k(g_z,g_i)^{-\Gamma^{\le n}(z,i)}\chi^k(h_z,g_i)^{\Gamma^{\le n}(z,i)}.
\end{equation*}

Using the characteristics of the function $\chi^k$ we can reduce the equation above to a single factor of the form $\chi^k(g^\prime,g_i)$, where 
\begin{equation*}
    \begin{aligned}
        g^\prime = \sum_{\{j_1,\ldots,j_{k-1}\} \in X \cup E} \Gamma^k_{i,j_1,\ldots,j_{k-1}}(h_{j_1} - g_{j_1}),\ldots,\\\sum_{\{j_1,\ldots,j_{k-1}\} \in X \cup E} \Gamma^k_{i,j_1,\ldots,j_{k-1}}(h_{j_{k-1}} - g_{j_{k-1}}).
    \end{aligned}
\end{equation*}
where the sum is taken over all possible $k-1$ elementary sets in $X \cup E$, and
\begin{equation*}
    \Gamma^k_{i_1,i_2,\ldots,i_k} =
    \begin{cases}
        1 & \textrm{if } i_1,i_2,\ldots,i_k \textrm{ are related},\\
        0 & \textrm{otherwise}.
    \end{cases}
\end{equation*}
Now the integral over $g_i$ gives us $\delta(g^\prime)$, which is same as:
\begin{equation}
\label{red2}
   |G|^X \quad \frac{\chi^{\Gamma^{\le n}}(h^{X \cup E}, h^{X \cup E})}{\chi^{\Gamma^{\le n}}(g^{X \cup E}, g^{X \cup E})}  \delta(H_{X \cup E}^I(h^{X \cup E} - g^{X \cup E})).
\end{equation}

The above equation can be reduced to the error correcting condition for an isometric function given in~\autoref{nsc}. As one factor of the equation is independent of the input vertices $X$ and the other is dependent factor of input vertices. And it is also to be noted that the above expression vanishes when $g^X = h^X$ and this implies that the $\delta$-function vanishes. So we can conclude that: if and only if $d^X \ne 0$, it implies $H_{X \cup E}^I(d^X) \ne 0$. 

Now assume that $H_{X \cup E}^I(d^X) = 0$ implies $d^X \ne 0$. The $\delta$-function can be reduced as:
\begin{equation*}
\begin{aligned}
    \delta(H_{X \cup E}^I(h^{X \cup E} - g^{X \cup E}))
    = \delta(H_{E}^I(h^{E} - g^{E}))\delta(h^{X} - g^{X}),
\end{aligned}
\end{equation*}

as two expressions are equal when $h^{X} = g^{X}$, and for $h^{X} \ne g^{X}$ both vanishes. Now for the bicharacter quotient in \autoref{red2}, we can easily simplify as follows:

\begin{equation*}
    \chi^{\Gamma^{\le n}}(h^{X}, h^{X})\chi^{\Gamma^{\le n}}(h^{E}, h^{X})\chi^{\Gamma^{\le n}}(h^{E}, h^{E}).
\end{equation*}
Similar decomposition holds for the denominator too. So the $(X,X)$ and $(E,E)$ factors cancel in numerator and denominator. So the remaining $(X,E)$ factors can be expressed as follows:
\begin{equation*}
    \frac{\chi^{\Gamma^{\le n}}(h^{E}, h^{X})}{\chi^{\Gamma^{\le n}}(g^{E}, h^{X})} = \prod_{k = 1}^{n}\prod_{j \in X} \chi^k(h_{j^\prime},h_j),
\end{equation*}
where,
\begin{equation*}
\begin{aligned}
    h_{j^\prime} = \{\sum_{\{j_1,\ldots,j_{k-1}\} \in X \cup E} \Gamma^k_{j_1,\ldots,j_{k-1},j}(h_{j_1} - g_{j_1}),\ldots,\\\sum_{\{j_1,\ldots,j_{k-1}\} \in X \cup E} \Gamma^k_{j_1,\ldots,j_{k-1}}(h_{j_{k-1},j} - g_{j_{k-1}})\}.
\end{aligned}
\end{equation*}

For the above equation to be in the desired form $C(g^E,h^E)$, this must be independent of all the $X$ variables. But this is satisfied, and the equation becomes independent of $h_j$ if and only if $j^\prime = 0$. Hence we must have $H_E^I(h^E - g^E) = 0$ implies $H_E^X(h^E - g^E) = 0$.

This concludes the proof.
\end{proof}

\subsection{Physical Encoding using Hypergraph States}

In this section we give two different ways to encode the error correcting code based on hypergraph states. The first one follows the secret sharing scheme using graph states as per~\cite{markham}. The encoding is given by:
\begin{equation*}
    \ket{0} \rightarrow \ket{0}_D\otimes \Gamma^{\le n}, \qquad \quad \ket{1} \rightarrow \ket{1}_D\otimes \overline{X}\Gamma^{\le n}, 
\end{equation*}
where $\overline{X} = X_1 \otimes X_2 \otimes \cdots \otimes X_{n}$. The subscript $D$ is to convey that the encoding is not using the physical bit in the process. Note that here one qubit is encoded into $n+1$ qubits.

The second way to encode the error correcting code uses the fact that a set of stabilizers can be defined on hypergraph states, as follows:
\begin{equation*}
    K_i = X_i \otimes \prod_{k=1}^{n} \prod_{\{v_1,\ldots,v_{k-1}\} \in N(i)} C^{k-1}Z_{v_1,\ldots,v_{k-1}}.
\end{equation*}
Based on the above definition, one can follow~\cite{cleve} for encoding and decoding using stabilizer formalism.

\section{Advantages Over Graph States}
We will now argue that the number of gate operations required to prepare a hypergraph state is less than that to prepare the corresponding graph state (for correcting equal number of errors) on the same topology. 

First we see how many standard Controlled-Z gates are required to implement one multi-qubit Controlled-Z gate. As per~\cite{marklov}, an $n$-qubit Controlled-Z gate requires $2n$ standard Controlled-Z gates.

To replicate any effects of the hypergraph states in the graph state, we need to form a complete graph with the vertices that form a hyper-edge in the corresponding hypergraph. So for a hyper-edge with $n$ nodes, we need a complete subgraph on $n$ nodes requiring $n(n-1)/2$ edges.

We can easily see that for $n \geq 6$, the number of Controlled-Z gates required for a complete graph would be greater than the number of Controlled-Z required to implement the multi-qubit Controlled-Z for the corresponding hypergraph state. For example, let us consider $n = 6$, the number of Controlled-Z required for a complete graph in this case 15. But we can implement a $C^6Z$ by using just 12 Controlled-Z.

\subsection{A Detailed Example Using 6-uniform Hypergraph}
\begin{figure}
    \centering
    \includegraphics[width=1\linewidth]{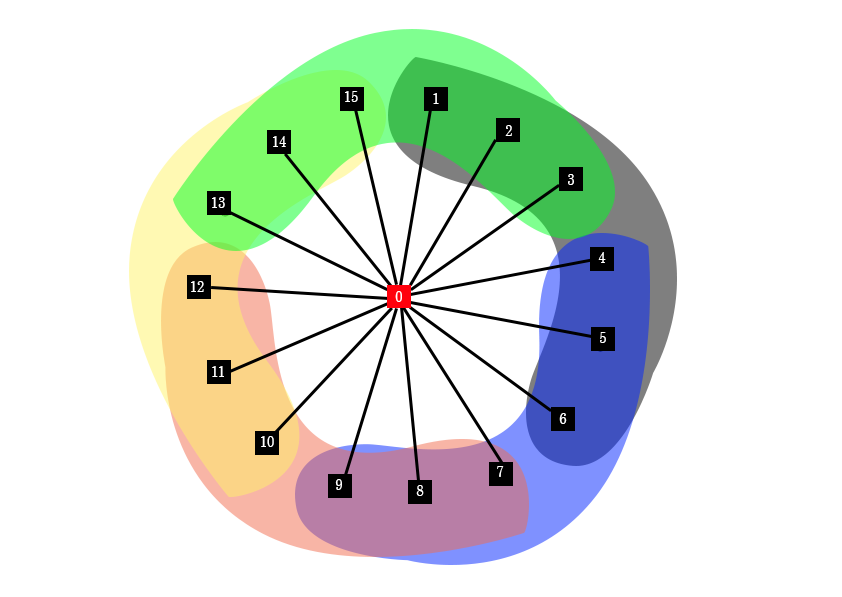}
    \caption{Hypergraph state for 15-fold way correcting 4 arbitrary errors}
    \label{fig:example}
\end{figure}

Now, we will give an example using hypergraphs. The hypergraph state for the code can be seen in~\autoref{fig:example}.

For this graph,
\begin{equation*}
    X = \{0\}, \qquad \qquad Y = \{1,2,\ldots,15\}.
\end{equation*}

$\Gamma$ for the graph is a 6-dimensional matrix, with zero entry everywhere except the following.
\begin{equation*}
    \begin{aligned}
        \Gamma(1,2,3,4,5,6) = 1,\\
        \Gamma(4,5,6,7,8,9) = 1,\\
        \Gamma(7,8,9,10,11,12) = 1,\\
        \Gamma(10,11,12,13,14,15) = 1,\\
        \Gamma(1,2,3,13,14,15) = 1.
    \end{aligned}
\end{equation*}

It is also to be noted that the adjacency matrix of the graph $\Gamma$ is a symmetric matrix.
The corresponding hypergraph state is given by:
\begin{equation*}
    \Gamma^4 = \prod_{\{z_1,z_2,z_3,z_4,z_5,z_6\}\in E} C^6Z \ket{+}^{\otimes15}.
\end{equation*}

Now we show here that, given a finite abelian group $G$, the hypergraph can detect up to four errors. There are six different error configurations and we will discuss it one by one. The vertices coloured red are the errors.

\begin{figure}[h]
    \centering
    \includegraphics[width=.8\linewidth]{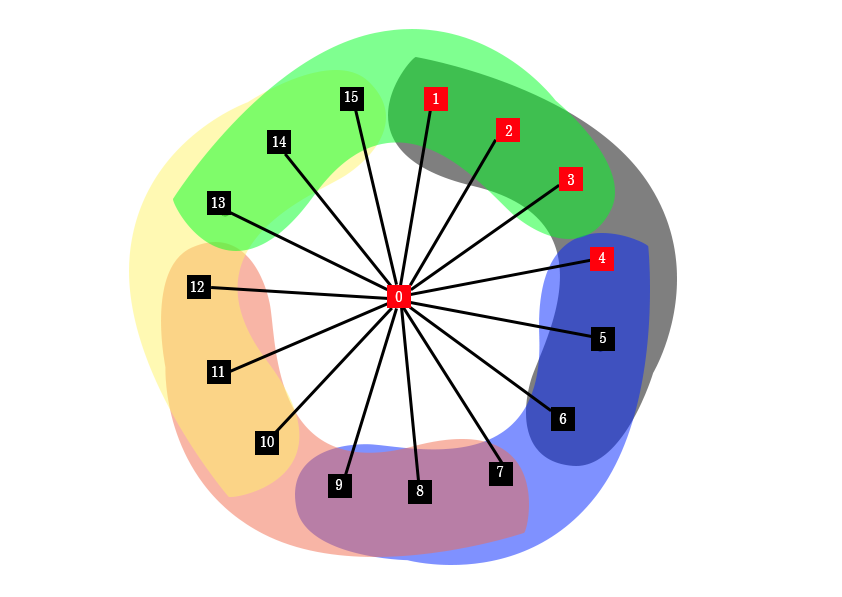}
    \caption{All errors occurring in same edge.\label{fig:e1}}
  \end{figure}
The first error configuration can be seen in \autoref{fig:e1} and the corresponding set of equations becomes as follows.
\begin{center}
\begin{tabular}{ |c|c| } 
 \hline
 Vertex $x$ & Equations \\ 
 \hline
 5 and 6 & $d_0 + d_1 + d_2 + d_3 + d_4 = 0$ \\ 
 7, 8 and 9 & $d_0 + d_4 = 0$ \\
 10, 11 and 12 & $d_0 = 0$\\
 13, 14 and 15 & $d_0 + d_1 + d_2 + d_3 = 0$\\
 \hline
\end{tabular}
\end{center}

Here $d_0 = 0$ which implies $d_1 + d_2 + d_3 + d_4 = 0$. Thus the condition in \autoref{suf2} is satisfied. So we can say this error configuration is detected.

\begin{figure}[h]
    \centering
    \includegraphics[width=.8\linewidth]{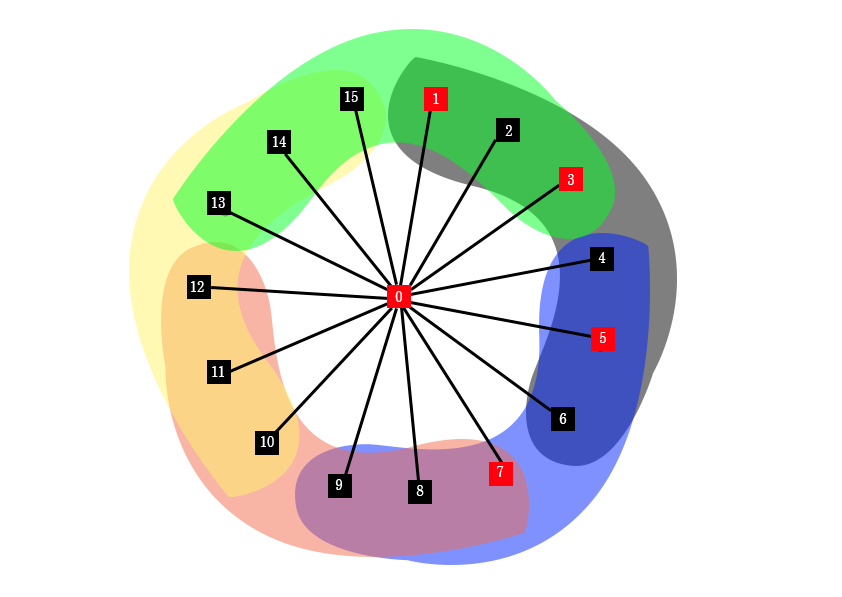}
    \caption{Errors occurring alternatively.\label{fig:e2}}
  \end{figure}
For the error configuration \autoref{fig:e2}, the set of equations becomes:
\begin{center}
\begin{tabular}{ |c|c| } 
 \hline
 Vertex $x$ & Equations \\ 
 \hline
 2 & $d_0 + d_1 + d_3 + d_5 = 0$ \\ 
 4 and 6 & $d_0 + d_1 + d_3 + d_5 + d_7 = 0$ \\
 8 and 9 & $d_0 + d_5 + d_7 = 0$\\
 10, 11 and 12 & $d_0 + d_7 = 0$\\
 13, 14 and 15 & $d_0 + d_1 + d_3 = 0$\\
 \hline
\end{tabular}
\end{center}

This set of equations clearly implies that $d_0 = d_5 = d_7 = 0$ and $d_1 + d_3 =0$. Thus the error configuration is detected.
\begin{figure}[h]
    \centering
    \includegraphics[width=.8\linewidth]{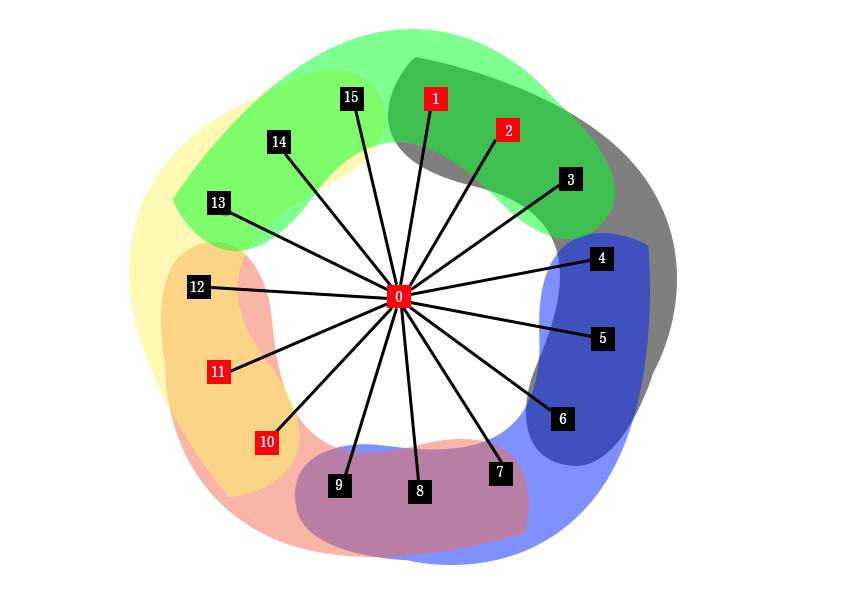}
    \caption{Two errors in same hyper-edge but not exactly opposite to each other.\label{fig:e3}}
  \end{figure}
For the error configuration \autoref{fig:e3}, the set of equations becomes:
\begin{center}
\begin{tabular}{ |c|c| } 
 \hline
 Vertex $x$ & Equations \\ 
 \hline
 3,4, 5 and 6 & $d_0 + d_1 + d_2 = 0$ \\ 
 7, 8, 9 and 12 & $d_0 + d_{10} + d_{11} = 0$ \\
 13, 14 and 15 & $d_0 + d_1 + d_2 + d_{10} + d_{11} = 0$\\
 \hline
\end{tabular}
\end{center}

Here $d_0 = 0$ and $d_1 + d_2 + d_{10} + d_{11} = 0$. Thus the condition in \autoref{suf2} is satisfied. So this error configuration is detected.

\begin{figure}[h]
    \centering
    \includegraphics[width=.8\linewidth]{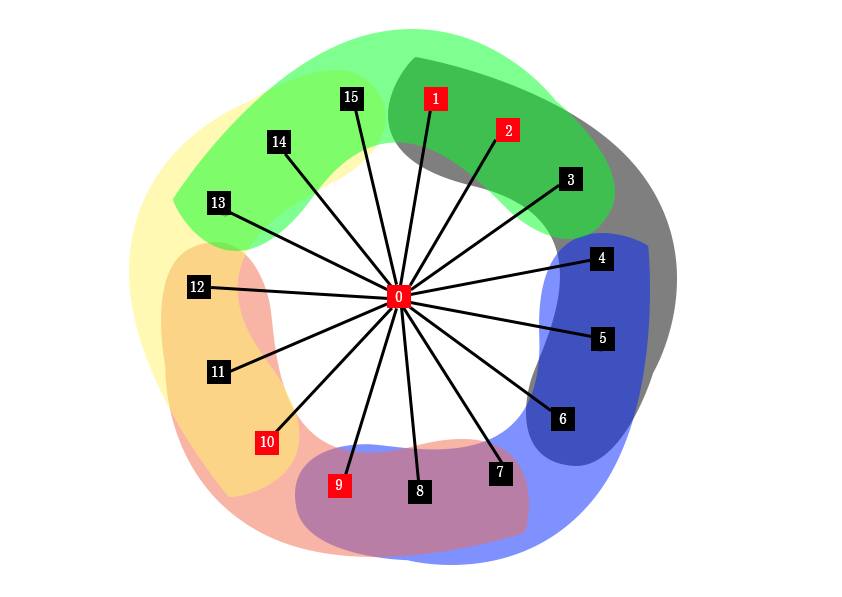}
    \caption{Two errors in same hyper-edge.\label{fig:e4}}
  \end{figure}
For the error configuration \autoref{fig:e4} the set of equations becomes:
\begin{center}
\begin{tabular}{ |c|c| } 
 \hline
 Vertex $x$ & Equations \\ 
 \hline
 3 & $d_0 + d_1 + d_2 = 0$ \\ 
 4, 5 and 6 & $d_0 + d_1 + d_2 + d_9 = 0$ \\
 7, 8, 11 and 12 & $d_0 + d_9 + d_{10} = 0$\\
 13, 14 and 15 & $d_0 + d_1 + d_2 + d_{10} = 0$\\
 \hline
\end{tabular}
\end{center}

here this set of equations implies that $d_0 = d_9 = d_{10} = 0$ and $d_1 + d_2 = 0$. Thus the error configuration is detected. 

\begin{figure}[h]
    \centering
    \includegraphics[width=.8\linewidth]{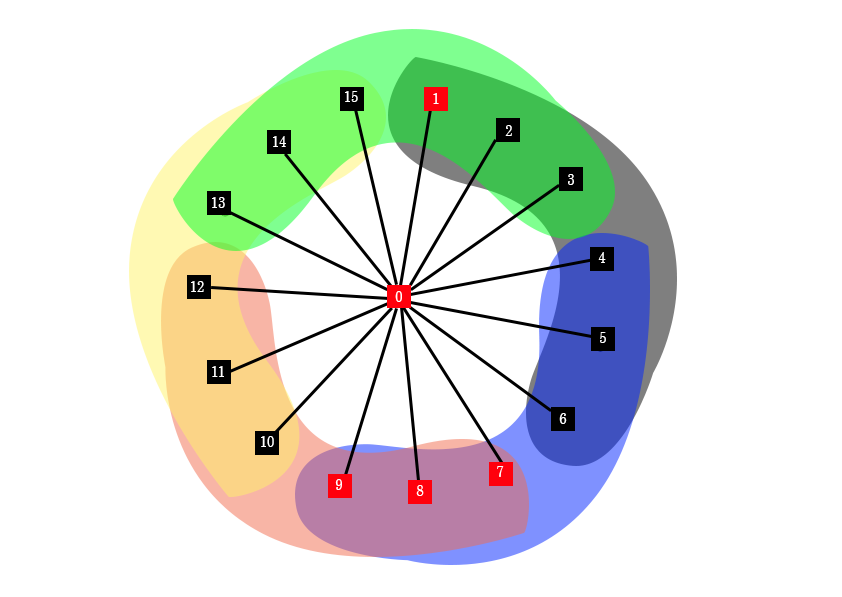}
    \caption{Three errors in different edges.\label{fig:e5}}
  \end{figure}
For the error configuration \autoref{fig:e5} the set of equations becomes:
\begin{center}
\begin{tabular}{ |c|c| } 
 \hline
 Vertex $x$ & Equations \\ 
 \hline
 2, 3, 13, 14 and 15 & $d_0 + d_1 = 0$ \\ 
 4, 5 and 6 & $d_0 + d_1 + d_7 + d_8 + d_9 = 0$ \\
 10, 11 and 12 & $d_0 + d_7 + d_8 + d_9 = 0$\\
 \hline
\end{tabular}
\end{center}

For the above set of equations, $d_0 = d_1 = 0$ and $d_7 + d_8 + d_9 = 0$. So this error configuration is also detected.

\begin{figure}[h]
    \centering
    \includegraphics[width=.8\linewidth]{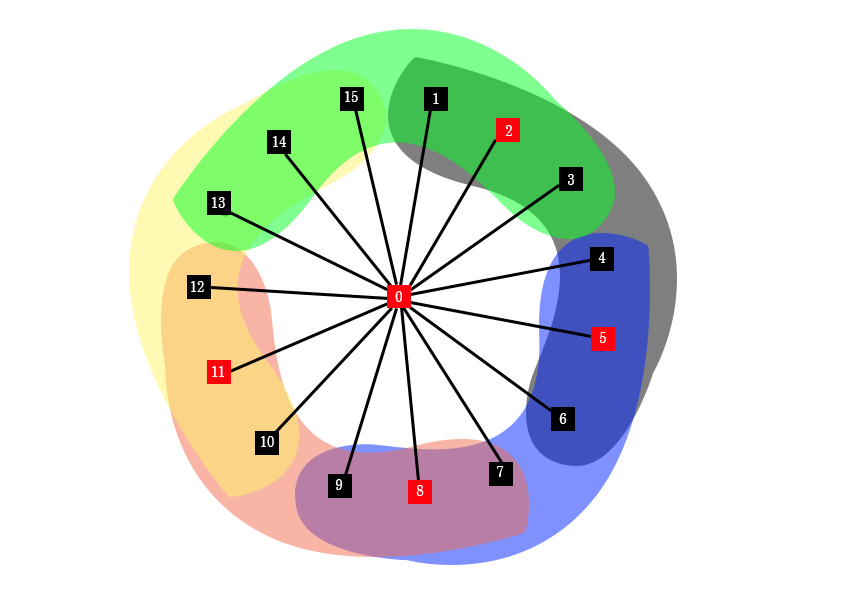}
    \caption{Errors occurring alternatively with distance 2 vertices apart.\label{fig:e6}}
  \end{figure}
For the error configuration \autoref{fig:e6}, the set of equations becomes:
\begin{center}
\begin{tabular}{ |c|c| } 
 \hline
 Vertex $x$ & Equations \\ 
 \hline
 1 and 3 & $d_0 + d_2 + d_5 = 0$ \\ 
 4 and 6 & $d_0 + d_2 + d_5 + d_8 = 0$ \\
 7 and 9 & $d_0 + d_5 + d_8 + d_{11} = 0$\\
 10 and 12 & $d_0 + d_8 + d_{11} = 0$\\
 13, 14 and 15 & $d_0 + d_2 + d_11 = 0$\\
 \hline
\end{tabular}
\end{center}

This set of equations clearly implies that $d_0 = d_2 = d_5 = d_8 = d_{11} = 0$. Thus the error configuration is detected. So we conclude by saying the code correct arbitrary four errors. 

\section{Conclusion}
In this paper we have discussed the error correcting properties of hypergraphs and arrived at a condition for error correction using the hypergraph states. We have also shown that to correct the equal number of errors on the same graph topology, the number of gate operations required to prepare the hypergraph states is less than that of the graph states for any hyper-edge of degree greater than or equal to 6.

Error correction using hypergraph states begins a new direction in the field of fault tolerant quantum computation. We believe further research in this area will reveal more insights and new interesting results.

\end{document}